\documentclass[12pt]{article}
{\begin{Sbox}\begin{minipage}}%
{\end{minipage}\end{Sbox}\fbox{\TheSbox}}%

\usepackage[psamsfonts]{amssymb}
\usepackage{amsmath, amsthm, anysize, enumerate, color, url}
\usepackage{graphicx}
\usepackage{epstopdf}
\usepackage{epsfig}
\usepackage{subfigure}
\usepackage[ruled,linesnumbered]{algorithm2e}
\usepackage[sectionbib]{natbib}

\usepackage{threeparttable}

\newtheorem{theorem}{Theorem} 

\newtheorem{lemma}{Lemma} 

\usepackage[colorlinks, breaklinks,
                   linkcolor=red,
                  citecolor=blue
                  ]{hyperref}

\usepackage{breakurl}

\newtheorem{definition}{Definition}

\theoremstyle{definition}
\newtheorem{remark}{Remark}  

\newcommand{\E}{\mathbb{E}}
\newcommand{\Z}{\mathbb{Z}}
\newcommand{\R}{\mathbb{R}}

\renewcommand{\P}{\mathbb{P}}

\usepackage{setspace}

\begin{document}

\title{Directed Networks with a Differentially Private Bi-degree Sequence\thanks{
We have changed the original title ``Directed networks with a noisy bi-degree sequence"
to ``Directed Networks with a Differentially Private Bi-degree Sequence" because contents are changed in the framework of differential privacy.}}

\author{Ting Yan\thanks{Department of Statistics, Central China Normal University, Wuhan, 430079, China.
\texttt{Email:} tingyanty@mail.ccnu.edu.cn.}
\\
Central China Normal University
}
\date{}

\maketitle

\begin{abstract}
Although a lot of approaches are developed to release
network data with a differentially privacy guarantee,
inference using noisy data in many network models is still
unknown or not properly explored.
In this paper, we release the bi-degree sequences of directed networks using
the Laplace mechanism and use the $p_0$ model for inferring the degree parameters.
We show that the estimator of the parameter without the denoised process is asymptotically consistent and normally distributed.
This is contrast sharply with some known results that valid inference such as the existence and consistency of the estimator needs the denoised process.
Along the way, a new phenomenon is revealed in which an additional variance factor
appears in the asymptotic variance of the estimator when the noise becomes large.
Further, we propose an efficient algorithm for finding the closet point lying in the set of all graphical bi-degree sequences under the global $L_1$ optimization
problem.
Numerical studies demonstrate our theoretical findings.


\end{abstract}

\vskip 5 pt


\section{Introduction}

As more and more network data (of all kinds, but especially social ones) are collected and made publicly available,
the data privacy has become an important issue in network data analysis since they may contain
sensitive information about individuals and their relationships (e.g., sexual relationships, email exchanges).
Directly publishing these sensitive data with anonymized or unanonymized nodes could cause severe privacy problems or even lead to legal actions.
For example, Netflix released the Netflix Prize data-set for public analysis in 2007, which contains anonymized network data
about the viewing habits of its members. Two years later, Netflix was
involved in a lawsuit with one of its members, who had been
victimized by the resulting privacy invasion, because the de-anonymization
technique make the re-identification of individuals possible using public information [\cite{Task:Clifton:2012}].
Nevertheless, the benefit to analyze them are obvious to addressing a
variety of important issues including disease
transmission, fraud detection, precision marketing, among many others.

To guarantee the confidence information not be disclosed,
sensitive network data must be carefully treated before being made public.
Although it is easy to attack under the anonymization technique by 
releasing an anonymized isomorphic network [e.g., \cite{Backstrom:Dwork:Kleinberg:2011}], 
some refined
anonymization techniques have been proposed [e.g., \cite{Campan:Truta:2009, Narayanan:Shmatikov:2009, Zhou:Pei:Luk:2008}].
These methods transform the original graph into a new graph by adding/removing edges or clustering of nodes into
groups. However, they depends on an
attacker's background knowledge and may fail to protect the private information.
\cite{Dwork:Mcsherry:Nissim:Smith:2006} develop a rigorous privacy standard for randomized data releasing mechanisms, differential privacy, to achieve privacy protection.
An algorithm that satisfies differential privacy, requires that the outputs should not be significantly different if the inputs are similar.
Differential privacy provides strong guarantees of privacy
without making any assumptions about the background knowledge of attackers.
Since then, it has been widely used as a privacy standard to release network data  [e.g., \cite{Hay:2009,Lu:Miklau:2014, Task:Clifton:2012, Jorgensen:Yu:Graham:2016}].

Although many differentially private algorithms have been developed to release network data or their
aggregate network statistics safely [e.g., \cite{Jorgensen:Yu:Graham:2016, Lu:Miklau:2014, Nguyen:Imine:Rusinowitch:2016, Task:Clifton:2012}],
statistical inference with noisy network data is still in its infancy.
How to accurately estimate model parameters and analyze asymptotic properties of their estimators  using noisy data
in many network models are still unknown or have not been properly explored.
There have been some recent developments in inference by using a differentially private degree sequence of undirected graphs.
\cite{Hay:2009} use the Laplace mechanism to release the degree partition and propose an efficient algorithm
to find the solution that minimizes the $L_2$ distance between all possible graphical degree partitions and the noisy degree partition.
With this post-processing step, they obtain an accurate estimate of the degree distribution of a graph.
\cite{Karwa:Slakovic:2016} use a discrete Laplace mechanism to release the degree sequence. By using the techniques for proving the consistency of the maximum likelihood estimator
 in the $\beta$-model in \cite{Chatterjee:Diaconis:Sly:2011} and those for obtaining its asymptotic normality in
\cite{Yan:Xu:2013}, \cite{Karwa:Slakovic:2016} proved that a differentially private estimator of the parameter in the $\beta$-model
is consistent and asymptotically normally distributed.
Moreover, they construct an efficient algorithm to denoise the differentially private degree sequence by solving an $L_1$ optimization problem.
\cite{Day:Li:Lyu:2016} proposed approaches based on aggregation and cumulative histogram to publish the degree distribution under node differential privacy.
\cite{Sealfon:Ullman:2019} proposed  an efficient algorithm for estimating the parameter of an Erd\"{o}s-R\'{e}nyi graph under node differential privacy.


In this paper, we focus on inference by using the differentially private bi-sequences of directed networks.
As pointed by \cite{Hay:2009},  it may fail to protect privacy if we directly release the degree sequence since
some graphs have unique degree sequences.
In some other scenarios, the bi-degrees of nodes are themselves sensitive information.
For instances, the out-degree of an individual reveals how many people are infected by him/her in
sexually transmitted disease networks and such information is sensitive.
In this case, it is required to limit disclosure of the bi-degrees.
We use the Laplace mechanism to release the bi-degree sequence and perform inference by using the noisy bi-sequence.
The main contributions are as follows. First, we show that the estimator of the parameter
in the $p_0$ model based on the moment equation in which the unobserved original bi-degree sequence is directly replaced by the noisy bi-sequence,
is consistent and asymptotically normal without
the denoised process. This is contrast sharply with some existing results [e.g., \cite{Fienberg:Rinaldo:Yang:2010, Karwa:Slakovic:2016}],
in which ignoring the noisy process can lead to non-consistent and even nonexistent parameter estimates.
The $p_0$ model is an exponential random graph model with the bi-degree sequence as its exclusively sufficient statistic.
Along the way, a new phenomenon is revealed in which an additional variance factor
appears in the asymptotic variance of the estimator when the noise becomes large. To the best of our knowledge, this is the first time
to discover this phenomenon in the noisy network data analysis.
We show that the differentially private estimator corresponding to the denoised bi-sequence
is also consistent and asymptotically normal.
Second, we propose an efficient algorithm to denoise the noisy bi-sequence,
which finds the closest point lying in the set of all possible graphical bi-degree sequences under the global $L_1$ optimization problem.
The denoised bi-sequence can be used to obtain an accurate estimate of the degree distribution of a directed graph.
Along the way, it also output a synthetic directed graph that can be used to infer the graph structure.
It is remarkable that the denoised step is needed in order to obtain valid estimates of graph structures since the noisy bi-sequence may not graphical.
Finally, we provide simulation studies as well as three real data analyses to illustrate theoretical results.

For the rest of the paper, we proceed as follows. In Section \ref{section:estimation}, we first
introduce a necessary  background on differential privacy and then present the estimation the degree parameter in the $p_0$ model using
the differentially private bi-sequence.
In Section \ref{section:asymptotic}, we present the consistency and asymptotic normality of the differentially private estimator.
In Section \ref{section:release}, we denoise the noisy bi-sequence, establish the upper bound of the error between the denoised bi-sequence and
the noisy bi-sequence and present the asymptotic properties of the estimator corresponding to the denoised bi-sequence.
In Section \ref{section:simulation}, we carry out the simulation studies to evaluate
the theoretical results and analyze three real network datasets.
We make the summary and further discussion in Section \ref{section:discussion}.
All proofs are regelated into online supplementary material.  

\section{Estimation from a differentially private bi-degree sequence}
\label{section:estimation}

Let $G_n$ be a simple directed graph on $n\geq 2$ nodes labeled by ``1, \ldots, n".
Here, ``simple" means that there are no multiple edges and no self-loops in $G_n$.
Let $A=(a_{i,j})$ be the adjacency matrix of $G_n$, where
$a_{i,j}$ is an indictor variable of the directed edge from head node $i$ to tail node $j$.
If there exists a directed edge from $i$ to $j$, then $a_{i,j}=1$; otherwise $a_{i,j}=0$.
Since $G_n$ is loopless, let $a_{i,i}=0$ for convenience.
Let $d_i^+= \sum_{j \neq i} a_{i,j}$ be the out-degree of node $i$
and $d^+=(d_1^+, \ldots, d_n^+)^\top$ be the out-degree sequence of the graph $G_n$.
Similarly, define $d_i^- = \sum_{j \neq i} a_{j,i}$ as the in-degree of node $i$
and $d^-=(d_1^-, \ldots, d_n^-)^\top$ as the in-degree sequence.
The pair $d=( (d^+)^\top, (d^-)^\top)^\top$ or $\{(d_1^+, d_1^-), \ldots, (d_n^+, d_n^-)\}$ are called the bi-degree sequence.

In this section,
we first give a brief introduction to differential privacy.
Then we release the bi-degree sequence under edge differential privacy and estimate the degree parameter in the $p_0$ model.

\subsection{Differential privacy}

Consider an original database $D$ containing a set of records of $n$ individuals.
We focus on mechanisms that take $D$ as input and output a sanitized database $S=(S_1, \ldots, S_k)$
for public use. The size of $S$ may not be the same as $D$.
A randomized data releasing mechanism $Q(\cdot|D)$ defines a conditional probability distribution on outputs $S$ given $D$.
Let $\epsilon$  be a positive real number and $\mathcal{S}$ denote the sample space of $Q$. The data releasing mechanism $Q$ is \emph{$\epsilon$-differentially private}
if for any two neighboring databases $D_1$ and $D_2$ that differ on a single element (i.e., the data of one person),
and all measurable subsets $B$ of $\mathcal{S}$ [\cite{Dwork:Mcsherry:Nissim:Smith:2006}],
\[
Q(S\in B |D_1) \leq e^{\epsilon }\times Q(S\in B |D_2).
\]

The privacy parameter $\epsilon$ is chosen by the data curator
administering the privacy policy and is public, which controls the trade-off
between privacy and utility.
Smaller value of $\epsilon$ means more privacy protection.

Differential privacy requires that the distribution of the output is almost the same whether or not an individual's record appears in the
database.  We illustrate why it protects privacy with an example. Suppose
a hospital wants to release some statistics on the medical records of their
patients to the public. In response,
a patient may wish to make his record omitted from
the study due to a privacy concern that the published results will
reveal something about him personally.
Differential privacy alleviates this concern because
whether or not the patient participates in the study, the
probability of a possible output is almost the same.
From a theoretical point, any test statistic has nearly no power for testing
whether an individual's data is in the original database or not;
see \cite{Wasserman:Zhou:2010} for a rigourous proof.

What is being protected in the differential privacy is precisely the difference between two
neighboring databases.
Within network data,
depending on the definition of
the graph neighbor, \emph{differential privacy} is divided into \emph{node differential privacy} [\cite{Kasiviswanathan:Nissim:Raskhodnikova:Smith;2013}] and
\emph{edge differential privacy} [\cite{Nissim:Raskhodnikova:Smith:2007}].
Two graphs are called neighbors if they differ in exactly one
edge, then \emph{differential privacy} is \emph{edge differential privacy}.
Analogously, we can define
\emph{node differential privacy} by letting graphs be neighbors if one can be obtained from the other by
removing a node and its adjacent edges.
Edge differential privacy protects edges not to be detected, whereas node differential privacy protects nodes together with their
adjacent edges, which is a stronger privacy policy.
However, it may be infeasible to design algorithms that are
both node differential privacy and have good utility. As an example, \cite{Hay:2009} show that estimating node degrees are
highly inaccurate under node differential privacy due to that the global sensitive in Definition \ref{definiton-2} is too large (in the worst case having an order $n$) such that
the output is useless.
Following \cite{Hay:2009}, we use edge differential privacy here.

Let $\delta(G, G^\prime)$ be the number of edges
on which $G$ and $G^\prime$ differ.
The formal definition of edge differential privacy is as follows.

\begin{definition}[Edge differential privacy]
Let $\epsilon>0$ be a privacy parameter. A randomized mechanism
$Q(\cdot |G)$ is $\epsilon$-edge differentially private if
\[
\sup_{ G, G^\prime \in \mathcal{G}, \delta(G, G^\prime)=1 } \sup_{ S\in \mathcal{S}}  \frac{ Q(S|G) }{ Q(S|G^\prime ) } \le e^\epsilon,
\]
where $\mathcal{G}$ is the set of all directed graphs of interest on $n$ nodes and
$\mathcal{S}$ is the set of all possible outputs.
\end{definition}

Let $f: \mathcal{G} \rightarrow \mathbb{R}^{k}$ be a function. The global sensitivity [\cite{Dwork:Mcsherry:Nissim:Smith:2006}] of the function $f$, denoted $\Delta f$, is defined below.

\begin{definition}
\label{definiton-2}
(Global Sensitivity).
Let $f:\mathcal{G} \to \R^k$. The global sensitivity of $f$ is defined as
\[
\Delta(f) = \max_{ \delta( G, G^\prime) =1 } \| f(G)- f(G^\prime) \|_1
\]
where $\| \cdot \|_1$ is the $L_1$ norm.
\end{definition}

The global sensitivity measures the worst case difference between any two neighboring graphs.
The magnitude of noises added in the differentially private algorithm $Q$ crucially depends on the global sensitivity.
If the outputs are the network statistics, then a simple algorithm to guarantee EDP is the Laplace Mechanism [e.g., \cite{Dwork:Mcsherry:Nissim:Smith:2006}]
that adds the Laplace noise proportional to the global sensitivity of $f$.

\begin{lemma}\label{lemma:DLM}(Laplace Mechanism).
Let $f:\mathcal{G} \to \R^k$. Let $e_1, \ldots, e_k$ be independent and identically distributed Laplace random variables with
density function $e^{-|x|/\lambda}/\lambda$.
Then the Laplace Mechanism outputs $f(G)+(e_1, \ldots, e_k)$ is $\epsilon$-edge differentially private, where $\epsilon= -\Delta(f)\log \lambda$.
\end{lemma}

When $f(G)$ is integer, one can use a discrete Laplace random variable as the noise as in \cite{Karwa:Slakovic:2016}, where it has the probability mass function:
\[
\P(X=x)= \frac{1-\lambda}{1+\lambda} \lambda^{|x|},~~x \in \{0, \pm 1, \ldots\}, \lambda\in(0,1).
\]
Lemma \ref{lemma:DLM} still holds if the continuous Laplace distribution is replaced by the discrete version.

One nice property of differential privacy is that any function of a differentially
private mechanism is also differentially private.

\begin{lemma}[\cite{Dwork:Mcsherry:Nissim:Smith:2006, Wasserman:Zhou:2010}]
\label{lemma:fg}
Let $f$ be an
output of an $\epsilon$-differentially private mechanism and $g$ be any function. Then
$g(f(G))$ is also $\epsilon$-differentially private.
\end{lemma}

By Lemma \ref{lemma:fg}, any post-processing done on the noisy
bi-degree sequences obtained as an output of a differentially private mechanism is also
differentially private.

\subsection{The differentially private bi-degree sequence}

We use the discrete Laplace mechanism in Lemma \ref{lemma:DLM} to
release the bi-degree sequence $d=(d^+, d^-)$ under edge differential privacy.
Note that $f(G_n)=(d^+, d^-)$.
If we add or remove a directed edge $i\to j$ in $G_n$, then the out-degree of the head node and
the in-degree of the tail node associated with the changed edge increase or decrease $1$ each.
Therefore, the global sensitivity for the bi-degree sequence is $2$.
The released steps are in Algorithm \ref{algorithm:a}, where a differentially private bi-sequence is returned.

\begin{algorithm}
\caption{Releasing $d$}
\label{algorithm:a}
\KwData{The bi-degree sequence $d$ and privacy parameter $\epsilon_n$}
\KwResult{The differentially private bi-sequence $z$}
Let $d=(d^+, d^-)$ be the bi-degree sequence of $G_n$\;
\For{ $i= 1 \to n$}{
 Generate two independent $e_i^+$ and $e_i^-$ from discrete Laplace with $\lambda_n= \exp(-\epsilon_n/2)$\;
 Let $z_i^+=d_i^+ + e_i^+$ and $z_i^- = d_i^- + e_i^-$
}
\end{algorithm}

\subsection{Estimation based on the $p_0$ model}
To make statistical inference from a noisy bi-sequence, we need to specify a model on the original bi-degree sequence.
If no prior information is given, we can model $d$ according to the maximum entropy principle [\cite{Wu:1997}].
It forces the probability distribution on graph $G_n$ into the exponential family with the bi-degree sequence as
the sufficient statistic, which admits the maximum entropy when the expectation of a bi-degree sequence is given.
Hereafter, we refer to this model as the $p_0$ model. The subscript ``0" means a simpler model than the $p_1$ model
that contains an additional reciprocity parameter [\cite{Holland:Leinhardt:1981}].
The $p_0$ model can be represented as:
\begin{equation}
\label{p0model}
\P(G_n)= \frac{1}{c(\alpha, \beta)} \exp( \sum_i \alpha_i d_i^+ + \sum_j \beta_j d_j^- ),
\end{equation}
where $c(\alpha, \beta)$ is a normalizing constant, $\alpha=(\alpha_1, \ldots, \alpha_n)^\top$
and $\beta=(\beta_1, \ldots, \beta_n)^\top$.
The outgoingness parameter $\alpha_i$ characterizes how attractive the node is and the incomingness parameter $\beta_{i}$ illustrate the extent to which the node is attracted to others as in \cite{Holland:Leinhardt:1981}.
Although the $p_0$ model looks simple, it is still useful to applications
where only the bi-degree sequence is used. First, it can be served as null
models for hypothesis testing [e.g., \cite{Holland:Leinhardt:1981,Fienberg:Wasserman:1981,Zhang:Chen:2013}]. Second, it can be
used to re-construct networks and make statistical inferences in a situation
in which only the bi-degree sequence is available due to privacy consideration
[e.g., \cite{Helleringer:Kohler:2007}]. Third, it can be used as a preliminary
analysis for choosing suitable statistics for network configurations
[e.g., \cite{Robins.et.al.2009}].

Since an out-edge from node $i$ pointing to $j$ is the in-edge of $j$ coming from $i$, it leads to that
the sum of out-degrees is equal to the sum of in-degrees.
If one transforms $(\alpha, \beta)$ to $(\alpha-c, \beta+c)$, the probability distribution in \eqref{p0model} does not change.
For the sake of the identification of model parameters, we set $\beta_n=0$ as in \cite{Yan:Leng:Zhu:2016}.
The $p_0$ model can be formulated by an array of mutually independent Bernoulli random variables $a_{i,j}$, $1\le i\neq j\le n$ with probabilities [\cite{Yan:Leng:Zhu:2016}]:
\[
\P(a_{i,j}=1) = \frac{ e^{\alpha_i + \beta_j} }{ 1 + e^{\alpha_i + \beta_j} }.
\]
The normalizing constant $c(\alpha, \beta)$ is $\sum_{i\neq j}\log( 1 + e^{\alpha_i+\beta_j} )$.
We use the following equations to estimate the degree parameter:
\renewcommand{\arraystretch}{1.2}
\begin{equation}\label{eq:likelihood-DP}
\large
\begin{array}{lcl}
z_i^+  & = & \sum_{j\neq i} \frac{e^{\alpha_i+\beta_j}}{1+e^{\alpha_i+\beta_j}}, ~~i=1, \ldots, n, \\
z_j^-  & = & \sum_{i\neq j} \frac{e^{\alpha_i+\beta_j}}{1+e^{\alpha_i+\beta_j}}, ~~j=1, \ldots, n-1,
\end{array}
\end{equation}
where $z$ is the differentially private bi-sequence of Algorithm \ref{algorithm:a}.
The fixed point iteration algorithm can be used to solve the above system of equations.
Since $E(e_i)=0$, the above equations are also the moment equations.
Let $\theta=(\alpha_1, \ldots, \alpha_n, \beta_1, \ldots, \beta_{n-1})^\top$.
The solution $\widehat{\theta}$ to the equations \eqref{eq:likelihood-DP} is the differentially private estimator of $\theta$ according to Lemma \ref{lemma:fg},
where $\widehat{\theta}=(\hat{\alpha}_1, \ldots, \hat{\alpha}_n, \hat{\beta}_1, \ldots, \hat{\beta}_{n-1} )^\top$
and $\hat{\beta}_n=0$.

\section{Asymptotic properties of the estimator}
\label{section:asymptotic}
In this section, we present the consistency and asymptotical normality of the differentially private estimator.
For a subset $C\subset \R^n$, let $C^0$ and $\overline{C}$ denote the interior and closure of $C$, respectively.  For a vector $x=(x_1, \ldots, x_n)^\top\in R^n$, denote by
$\|x\|_\infty = \max_{1\le i\le n} |x_i|$, the $\ell_\infty$-norm of $x$.  For an $n\times n$ matrix $J=(J_{i,j})$, let $\|J\|_\infty$ denote the matrix norm induced by the $\ell_\infty$-norm on vectors in $\R^n$, i.e.
\[
\|J\|_\infty = \max_{x\neq 0} \frac{ \|Jx\|_\infty }{\|x\|_\infty}
=\max_{1\le i\le n}\sum_{j=1}^n |J_{i,j}|.
\]

Since the number of parameters increase with the number of nodes, classical statistical theories can not be directly applied to obtain the asymptotic results of
estimator. We use the Newton method developed in \cite{Yan:Leng:Zhu:2016} to show the consistency. Here we deal with  not only the high dimensional problem but also the
errors carried by noises while \cite{Yan:Leng:Zhu:2016} only considered the high dimensional problem.
The idea of the proof for the existence and consistency of $\widehat{\theta}$ can be briefly described as follows.
Define a system of functions:
\renewcommand{\arraystretch}{1.2}
\begin{equation}\label{eq:F-DP}
\large
\begin{array}{lll}
F_i( \theta ) &  =  & z_i^+ -  \sum_{k=1; k \neq i}^n \frac{e^{\alpha_i+\beta_k}}{1+e^{\alpha_i+\beta_k} }, ~~~  i=1, \ldots, n, \\
F_{n+j}( \theta ) & = & z_j^- - \sum_{k=1; k\neq j}^n \frac{e^{\alpha_k+\beta_j}}{1+e^{\alpha_k+\beta_j}},  ~~~  j=1, \ldots, n, \\
F( \theta ) & = & (F_1( \theta ), \ldots, F_{2n-1}( \theta ))^\top.
\end{array}
\end{equation}
Note the solution to the equation $F( \theta )=0$ is precisely the estimator.
We construct the Newton iterative sequence: $\theta^{(k+1)} = \theta^{(k)} - [ F'(\theta^{(k)})]^{-1} F(\theta^{(k)})$. If the initial value is chosen as
the true value $\theta^*$, then it is left to bound the error between the initial point
and the limiting point to show the consistency. This is done by establishing a geometric convergence of rate for the iterative sequence.
The details are in online supplementary material.
The existence and consistency of $\widehat{\theta}$ is stated blow.

\begin{theorem}\label{Theorem:con}
Assume that $A \sim \P_{ \theta^*}$, where $\P_{ \theta^*}$ denotes
the probability distribution \eqref{p0model} on $A$ under the parameter $\theta^*$.
If $(1+ 4\epsilon_n^{-1}) e^{12\|\theta^*\|_\infty } = o( (n/\log n)^{1/2} )$,
 then with probability approaching one as $n$ goes to infinity, the estimator $\widehat{\theta}$ exists
and satisfies
\[
\|\widehat{\theta} - \theta^* \|_\infty = O_p\left(  (1+\frac{4}{\epsilon_n}) \frac{ (\log n)^{1/2}e^{6\|\theta^*\|_\infty} }{ n^{1/2} } \right)=o_p(1).
\]
Further, if $\widehat{\theta}$ exists, it is unique.
\end{theorem}

\begin{remark}
The condition $(1+\frac{4}{\epsilon_n}) e^{12\|\theta^*\|_\infty } = o( (n/\log n)^{1/2} )$ in Theorem \ref{Theorem:con} to guarantee the consistency of the estimator,
exhibits an interesting trade-off between the privacy parameter $\epsilon_n$ and $\|\theta^*\|_\infty$. If $\|\theta^*\|_\infty$ is bounded by a constant, $\epsilon_n$ can be as small as
$n^{1/2}/(\log n)^{-1/2}$. Conversely, if $e^{\|\theta^*\|_\infty}$ is growing at a rate of $n^{1/12}/(\log n)^{1/12}$,
then $\epsilon_n$ can only be at a constant magnitude.
\end{remark}

In order to present asymptotic normality of $\widehat{\theta}$, we introduce a class of matrices.
Given two positive numbers $m$ and $M$ with $M \ge m >0$, we say the $(2n-1)\times (2n-1)$ matrix $V=(v_{i,j})$ belongs to the class $\mathcal{L}_{n}(m, M)$ if the following holds:
\begin{equation}\label{eq:LmM}
\begin{array}{l}
m\le v_{i,i}-\sum_{j=n+1}^{2n-1} v_{i,j} \le M, ~~ i=1,\ldots, n-1; ~~~ v_{n,n}=\sum_{j=n+1}^{2n-1} v_{n,j}, \\
v_{i,j}=0, ~~ i,j=1,\ldots,n,~ i\neq j, \\
v_{i,j}=0, ~~ i,j=n+1, \ldots, 2n-1,~ i\neq j,\\
m\le v_{i,j}=v_{j,i} \le M, ~~ i=1,\ldots, n,~ j=n+1,\ldots, 2n-1,~ j\neq n+i, \\
v_{i,n+i}=v_{n+i,i}=0,~~ i=1,\ldots,n-1,\\
v_{i,i}= \sum_{k=1}^n v_{k,i}=\sum_{k=1}^n v_{i,k}, ~~ i=n+1, \ldots, 2n-1.
\end{array}
\end{equation}
Clearly, if $V\in \mathcal{L}_{n}(m, M)$, then $V$ is a $(2n-1)\times (2n-1)$ diagonally dominant, symmetric nonnegative
matrix.
Define $v_{2n,i}=v_{i,2n}:= v_{i,i}-\sum_{j=1;j\neq i}^{2n-1} v_{i,j}$ for $i=1,\ldots, 2n-1$ and $v_{2n,2n}=\sum_{i=1}^{2n-1} v_{2n,i}$.
\cite{Yan:Leng:Zhu:2016} propose to approximate the inverse of $V$, $V^{-1}$, by the matrix $S=(s_{i,j})$, which is defined as
\begin{equation}
\label{definition:S}
s_{i,j}=\left\{\begin{array}{ll}\frac{\delta_{i,j}}{v_{i,i}} + \frac{1}{v_{2n,2n}}, & i,j=1,\ldots,n, \\
-\frac{1}{v_{2n,2n}}, & i=1,\ldots, n,~~ j=n+1,\ldots,2n-1, \\
-\frac{1}{v_{2n,2n}}, & i=n+1,\ldots,2n-1,~~ j=1,\ldots,n, \\
\frac{\delta_{i,j}}{v_{i,i}}+\frac{1}{v_{2n,2n}}, & i,j=n+1,\ldots, 2n-1,
\end{array}
\right.
\end{equation}
where $\delta_{i,j}=1$ when $i=j$ and $\delta_{i,j}=0$ when $i\neq j$.

We use $V$ to denote the Fisher information matrix of $\theta$ in the $p_0$ model.
It can be shown that
\[
v_{ij}=\frac{ e^{\alpha_i + \beta_j} }{ (1 + e^{\alpha_i + \beta_j})^2 },~~ 1\le i\neq j \le n.
\]
Since $e^x/(1+e^x)^2$ is an increasing function on $x$ when $x\ge 0$ and
a decreasing function when $x\le 0$, we have
\[
\frac{(n-1)e^{2\| \theta\|_\infty}}{(1+e^{2\| \theta\|_\infty})^2}\le v_{ii} \le \frac{n-1}{4}, ~~ i=1, \ldots, 2n.
\]
Therefore $V\in \mathcal{L}_n(m,M)$, where $m$ is the left expression and $M$ is the right expression in the above inequality.
The asymptotic distribution of $\widehat{\theta}$ depends on $V$.
Let $g=(d_1^+, \ldots, d_n^+, d_1^-, \ldots, d_{n-1}^-)^\top$ and
$\tilde{g}=(z_1^+, \ldots, z_n^+, z_1^-, \ldots, z_{n-1}^-)^\top$.
If we apply Taylor's expansion to each component of $\tilde{g} -\E g$, then the second order term in the expansion is $V(\widehat{\theta} - \theta)$.
Since $V^{-1}$ does not have a closed form, we work with $S$ defined at \eqref{definition:S} to approximate it. Then we represent $\widehat{\theta}-\theta$ as the sum of
$S(\tilde{g} - \E g)$ and a remainder. The central limit theorem is proved by establishing the asymptotic normality of $S( \tilde{g} - \E g)$ and
showing the remainder is negligible.
We formally state the central limit theorem as follows.

\begin{theorem}\label{Theorem:binary:central}
Assume that $A\sim \P_{\theta^*}$ and $(1+\frac{4}{\epsilon_n})^2 e^{ 18\|\theta^*\|_\infty } = o( (n/\log n)^{1/2} )$.\\
(i) If $\frac{4}{\epsilon_n} (\log n)^{1/2} e^{2\|\theta^*\|_\infty}=o(1)$,
then for any fixed $k\ge 1$, as $n \to\infty$, the vector consisting of the first $k$ elements of $(\widehat{\theta}-\theta^*)$ is asymptotically multivariate normal with mean $\mathbf{0}$ and covariance matrix given by the upper left $k \times k$ block of $S$ defined at \eqref{definition:S}.\\
(ii) Let
\[
s_n^2=\mathrm{Var}( \sum_{i=1}^n e_i^+ - \sum_{i=1}^{n-1} e_i^-)=2(2n-1) \frac{ e^{-\epsilon_n/2 } }{ ( 1 - e^{ -\epsilon_n/2} )^2 }.
\]
If  $s_n/v_{2n,2n}^{1/2} \to c$ for some constant $c$,
then for any fixed $k \ge 1$, the vector consisting of the first $k$ elements of $(\widehat{\theta}-\theta^*)$ is asymptotically  $k$-dimensional multivariate normal distribution with mean $\mathbf{0}$
and covariance matrix
\[
\mathrm{diag}( \frac{1}{v_{1,1}}, \ldots, \frac{1}{v_{k,k}})+ (\frac{1}{v_{2n,2n}} + \frac{s_n^2}{v_{2n,2n}^2}) \mathbf{1}_k \mathbf{1}_k^\top,
\]
where $\mathbf{1}_k$ is a $k$-dimensional column vector with all entries $1$.
\end{theorem}

\begin{remark}
First, if we change the first $k$ elements of $(\widehat{\theta}-\theta^*)$ to an arbitrarily fixed $k$ elements with the subscript set $\{i_1, \ldots, i_k \}$,
Theorem \ref{Theorem:binary:central} still holds. This is because all steps in the proof are valid if we change the first $k$ subscript set $\{1, \ldots, k\}$ to $\{i_1, \ldots, i_k \}$.
Second, the asymptotic variance for the difference of the pairwise estimators $(\widehat{\theta}-\theta^*)_i - (\widehat{\theta}-\theta^*)_j$ is $1/v_{i,i} + 1/v_{j,j}$, regardless of
the additional variance factor $1/v_{2n,2n} + s_n^2/v_{2n,2n}^2$.
\end{remark}

\begin{remark}
In the second part of Theorem \ref{Theorem:binary:central}, the asymptotic variance of $\widehat{\theta}_i$ has an additional variance factor $s_n^2/v_{2n,2n}^2$.
This is different from Theorem 2 in \cite{Yan:Leng:Zhu:2016}, in which they consider the a non-differential private case.
The asymptotic expression of $\hat{\theta}_i$ contains a term $\sum_{i=1}^n e_i^+ - \sum_{i=1}^{n-1} e_i^-$. Its variance is in the magnitude of $n  e^{-\epsilon_n/2 }$.
When  $\epsilon_n$ becomes small, the variance increases quickly and its impact on the $\widehat{\theta}_i$ can not be ignored when it increases to a certain level.
This leads to the appearance of the additional variance factor.

\end{remark}

\section{The denoised bi-degrees and synthetic directed graphs}
\label{section:release}

The output $z$ of Algorithm \ref{algorithm:a} generally is not the graphical bi-degree sequence.
There have been several characterizations for the bi-degree sequence [e.g., \cite{Fulkerson:1960,Kleitman:Wang:1973,Majcher:1985}].
A necessary condition for graphical bi-degree sequences is that the sum of in-degrees is equal to
that of out-degrees and all in- and out- degrees are between $0$ and $n-1$.
To check what are the chances that this condition holds, we carry out some simulations.
We use the $p_0$ model to generate the random graphs and record their bi-degree sequences. Then use Algorithm \ref{algorithm:a}
to output the bi-sequence $z$. We set $\alpha_i, \beta_i \sim U(0,1)$ and $n=100$. We repeat $10,000$ simulations and record
the frequency that $\sum_i z_i^+=\sum_i z_i^-$ holds. The simulation results show that this condition holds with at most $1\%$.

To make $z$ be graphical, we need to denoise $z$.
The denoising process appears to be complex.
First, the number of parameters to be estimated $(d_i^+, d_i^-, i = 1, \ldots, n)$ is equal to the number of observations
$(z_i^+, z_i^-, i = 1, \ldots, n)$.
Second, the parameter space is discrete and very large, whose cardinality grows at least an exponential magnitude.
Let $B_n$ be the set of all possible bi-degree sequence of graph $G_n$.
It is natural to use the closest point $\hat{d}$ lying in $B_n$ as the denoised bi-sequence with some distance between $\hat{d}$ and $d$.
We use $L_1$ distance here and define the estimator as
\begin{equation}
\label{eq:optimize}
\hat{d}=\arg\min_{ d\in B_n}  ( \| z^+ - d^+ \|_1 + \| z^- - d^- \|_1).
\end{equation}
Notice that the maximum likelihood estimation leads to the same solution.
Specifically, since the parameter $\lambda_n$ in the noise addition process of Algorithm \ref{algorithm:a} is known,
the likelihood on observation $z$ with the parameter $d$ in $B_n$ is
\[
L( d | z )= c(\lambda_n) \exp \{ - (\sum_{i=1}^n |z_i^+ - d_i^+| + \sum_{i=1}^{n-1} |z_i^- - d_i^-|) \}.
\]
We can see that the MLE of $d$ is also $\hat{d}$.

We propose Algorithm \ref{algorithm:b} to produce the MLE $\hat{d}$.
Along the way, it also outputs a directed graph with $\hat{d}$ as its bi-degree sequence.
The correctness of Algorithm \ref{algorithm:b} is given in Theorem \ref{thm:MLE:degree}, whose proof is in online supplementary material.

\begin{algorithm}[!htp]
\caption{Denoising $z$ }
\label{algorithm:b}
\KwData{A bi-sequence of integers $z=(z^+, z^-)$}
\KwResult{A directed graph $G_n$ on $n$ vertices with bi-degree sequence $\hat{d}$}
Let $G_n$ be the empty graph on $n$ vertices\;
Let $S=\{1, \ldots, n\} \setminus \{i: z_i^+\le 0\}$\;
\While {$|S|>0$}
{
     $T = \{1, \ldots, n\} \setminus \{i: z_i^- \le 0\}$\;
     Let $z_{i^*}^+ = \max_{i \in S} z_i^+$ and $i^* = \min \{i\in S: z_i^+ = z_{i^*}^+ \}$\;
     Let $T = T \setminus \{i^*\}$ and $pos=|T|$\;
     Let $h_{i^*}= \min( z_{i^*}^+, pos)$\;
     Let $I=$indices of $h_{i^*}$ highest values in $z^-(T)$ where $z^-(T)$ is the sequence $z^-$\;
     restricted to the index set $T$\;
     Add a directed edge from $i^*$ to $k$ in $G_n$ for each $k\in I$\;
     Let $z_i^- = z_i^- - 1$ for all $i\in I$ and $S = S\setminus \{i^*\}$
}
\end{algorithm}

\begin{theorem}
\label{thm:MLE:degree}
Let $z=(z^+, z^-)$ be a bi-sequence of integers obtained from Algorithm \ref{algorithm:a}.
The bi-degree sequence of $G_n$ produced by Algorithm \ref{algorithm:b} is $\hat{d}$ defined at \eqref{eq:optimize}.
\end{theorem}

We prove Theorem \ref{thm:MLE:degree} by converting the directed Havel-Hakimi algorithm [\cite{Erdos:Peter:Miklos:Toroczkai:2010}]
into Algorithm \ref{algorithm:b} that performs $L_1$ ``projection" on the set
$B_n$, which motivated by \cite{Karwa:Slakovic:2016} who use the Havel-Hakimi algorithm [\cite{Havel:1955,Hakimi:1962}]
to find the solution to the undirected $L_1$ optimalization problem. Although the Havel-Hakimi algorithm had been proposed sixty years ago,
the directed version has been derived until \cite{Erdos:Peter:Miklos:Toroczkai:2010}.
In the directed case, one needs to consider the in-degree sequence and out-degree sequence
simultaneously. Therefore, our algorithm is not a trivial extension from the algorithm in the undirected case in \cite{Karwa:Slakovic:2016}.



\begin{remark}
In step 8 of Algorithm \ref{algorithm:b}, if some in-degrees of $z^{-}(T)$ are equal, we arrange
them by the decreasing order of their corresponding out-degrees. Assume that the order is
$z^-_{i_1}\ge \cdots \ge z^-_{i_k}$. Then we select their top $h_{i^*}$ values.
This rule applies hereafter and we will not emphasize it.
\end{remark}

The next theorem characterizes the error between $\hat{d}$ and $d$ in terms of the privacy parameter $\epsilon_n$.

\begin{theorem}
\label{lemma:con:3}
When $\epsilon_n (c+1) \ge 4 \log n$, we have
\[
\P( \| \hat{d} - d\|_\infty  > c ) \le \frac{4}{n},
\]
where for two bi-sequences $a=(a^+, a^-)$ and $b=(b^+, b^-)$, $\|a - b\|_\infty$ is defined as
\begin{equation}
\label{definition:ab}
\| a - b \|_\infty = \max\{ \| a^+ - b^+\|_\infty, \|a^- - b^- \|_\infty \}
\end{equation}
\end{theorem}

As expected, the privacy parameter $\epsilon_n$ is smaller, the error between the original bi-degree and
its MLE $\hat{d}$ becomes larger. For any fixed $\tau\in (0, 1/2)$, if $\epsilon_n=\Omega(n^{-(1/2-\tau)})$, then
\begin{equation}
\label{eq:epsilon}
\| \hat{d} - d \|_\infty = O_p( n^{(1/2-\tau)}\log n).
\end{equation}

Both $\tilde{d}$ and $\hat{d}$ are the EDP estimator of $d$, 
where the latter is due to Lemma \ref{lemma:fg}.
We can use $\hat{d}$ to replace $\tilde{d}$ in equations \eqref{eq:likelihood-DP}
to obtain the denoised estimator of the parameter $\theta$ and denote the solution as $\bar{\theta}$.
By repeatedly using Lemma \ref{lemma:fg}, $\widehat{\theta}$ and $\bar{\theta}$ are both EDP estimators.
By noting \eqref{eq:epsilon} holds, with the similar lines of arguments for Theorems \ref{Theorem:con} and \ref{Theorem:binary:central},
the DP estimator is consistent and asymptotically normal stated in Theorem \ref{Theorem:DP}, whose proof is omitted.

\begin{theorem}\label{Theorem:DP}
Assume that $A \sim \P_{ \theta^*}$. \\  
(i) If $e^{12\| \theta^*\|_\infty}=o( (n/\log n)^{1/2})$ and $\epsilon_n=\Omega( (\log n/n)^{1/2})$, then as $n$ goes to infinity,
with probability approaching one, the EDP estimator $\bar{\theta}$ exists
and satisfies
\[
\|\bar{\theta} - \theta^* \|_\infty = O_p\left( \frac{ (\log n)^{1/2}e^{6\|\theta^*\|_\infty} }{ n^{1/2} } \right)=o_p(1).
\]
Further, if $\bar{\theta}$ exists, it is unique.\\
(ii) If  $e^{18\|\theta^*\|_\infty} =o( (n/\log n)^{1/2})$ and $\epsilon_n^{-1}e^{6\|\theta^*\|_\infty} = o(n^{1/2}/\log n)$, then for any fixed $k\ge 1$, as $n \to\infty$, the vector consisting of the first $k$ elements of $(\bar{\theta}-\theta^*)$ is asymptotically multivariate normal with mean $\mathbf{0}$ and covariance matrix given by the upper left $k \times k$ block of $S$ defined at \eqref{definition:S}.
\end{theorem}

\begin{remark}
Since the distribution of the difference $\hat{d} - d$ is difficult to obtain, we don't have the asymptotic result like in Theorem \ref{Theorem:binary:central} (ii).
By Theorem \ref{Theorem:DP}, the convergence rate of $\bar{\theta}_i$ is $1/v_{i,i}^{1/2}$ for any fixed $i$. Since $(n-1)e^{-2\|\theta^*\|_\infty}/4\le v_{i,i}\le (n-1)/4$, the rate of convergence is between $O(n^{-1/2}e^{\|\theta^*\|_\infty})$ and $O(n^{-1/2})$,
which is the same as the non private estimator [\cite{Yan:Leng:Zhu:2016}].
\end{remark}

\section{Numerical studies}
\label{section:simulation}

\subsection{Simulation}

In this section, we carry out numerical simulations by using the discrete Laplace mechanism in Algorithm \ref{algorithm:a}. 
We assess the performance of the estimator for finite sizes of networks when
$n$, $\epsilon_n$ or the range of $\theta_i$ varies and compare the simulation results of the non-denoised estimator with those of the denoised estimator.

The parameters in the simulations are as follows.
Similar to \cite{Yan:Leng:Zhu:2016}, the setting of the parameter $\theta^*$ takes a linear form.
Specifically, we set $\alpha_{i+1}^* = (n-1-i)L/(n-1)$ for $i=0, \ldots, n-1$.
For the parameter values of $\beta$, let $\beta_i^*=\alpha_i^*$, $i=1, \ldots, n-1$ for simplicity and $\beta_n^*=0$ by default.
We considered four different values for $L$, $L=0$, $\log(\log n)$, $(\log n)^{1/2}$ and $\log n$, respectively.
We simulated three different values for $\epsilon_n$: one is fixed ($\epsilon_n=2$) and the other two values tend to zero with $n$, i.e., $\epsilon_n=\log (n)/n^{1/4}, \log(n)/n^{1/2}$.
We considered three values for $n$, $n=100, 200$ and $500$.
Each simulation was repeated $10,000$ times.

By Theorem \ref{Theorem:binary:central}, $\hat{\xi}_{i,j} = [\hat{\alpha}_i-\hat{\alpha}_j-(\alpha_i^*-\alpha_j^*)]/(1/\hat{v}_{i,i}+1/\hat{v}_{j,j})^{1/2}$, $\hat{\zeta}_{i,j} = (\hat{\alpha}_i+\hat{\beta}_j-\alpha_i^*-\beta_j^*)/(1/\hat{v}_{i,i}+1/\hat{v}_{n+j,n+j})^{1/2}$, and $\hat{\eta}_{i,j} = [\hat{\beta}_i-\hat{\beta}_j-(\beta_i^*-\beta_j^*)]/(1/\hat{v}_{n+i,n+i}+1/\hat{v}_{n+j,n+j})^{1/2}$
converge in distribution to the standard normal distributions, where $\hat{v}_{i,i}$ is the estimate of $v_{i,i}$ by replacing $\theta^*$ with $\widehat{\theta}$.  Therefore, we assess the asymptotic normality of $\hat{\xi}_{i,j}$, $\hat{\zeta}_{i,j}$ and $\hat{\eta}_{i,j}$ using the quantile-quantile (QQ) plot.  Further, we record the coverage probability of the $95\%$ confidence interval, the length of the confidence interval, and the frequency that the estimate does not exist.  The results for $\hat{\xi}_{i,j}$, $\hat{\zeta}_{i,j}$ and $\hat{\eta}_{i,j}$ are similar, thus only the results of $\hat{\xi}_{i,j}$ are reported.
Note that $\bar{\theta}$ denotes the denoised estimator corresponding to the denoised bi-degree sequence $\hat{d}$.
The notation $\bar{\xi}_{i,j}$ is similarly defined and it also has the same asymptotic distribution as $\hat{\xi}_{i,j}$ by Theorem \ref{Theorem:DP}.
We also draw the QQ plots for $\bar{\xi}_{i,j}$ and $\hat{\alpha}_i-\alpha_i^*$. The distance between the original bi-degree sequence $d$ and the noisy bi-sequence  $z$ is also reported
in terms of $\|d - z\|_\infty$.


The average value of the $\ell_\infty$-distance between $d$ and $z$ is reported in Table \ref{Table:b}.
We can see that the distance becomes larger as $\epsilon_n$ decreases. It means that smaller $\epsilon_n$ provides more privacy protection.
For example, when $\epsilon_n$ changes from $\log n/n^{1/4}$ to $\log n/n^{1/2}$, $\|d-z\|_\infty$ dramatically increases from $8$ to $26$ in the case $n=100$.
As expected,  the distance also becomes larger as $n$ increases when $\epsilon_n$ is fixed.

\begin{table}[h]
\centering
\caption{The distance $\|d-z\|_\infty$. }
\label{Table:b}
\vskip5pt
\begin{tabular}{cccc}
\hline
         & \multicolumn{3}{c}{$\epsilon_n$} \\
         \cline{2-4}
$n$      & $2$ & $\log n/n^{1/4}$ & $\log n/n^{1/2}$  \\
\hline
$100$      & $5.7$     & $8.0$   & $25.5$ \\
$200$      & $6.4$     & $9.2$   & $35.1$ \\
$500$      & $7.4$     & $11.3$  & $53.8$ \\
\hline
\end{tabular}
\end{table}

\begin{figure}[!htb]
\centering
\includegraphics[ height=5in, width=6in, angle=0]{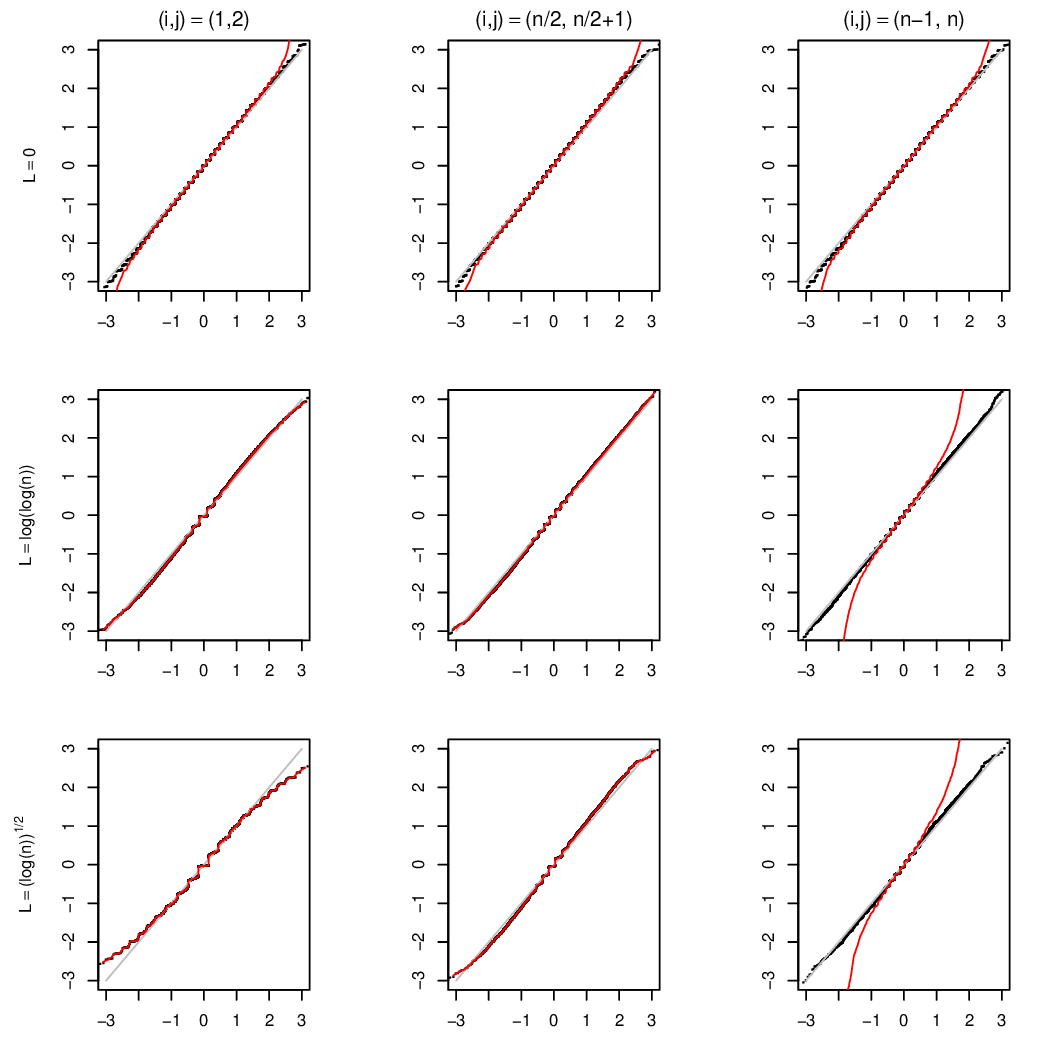}
\caption{The QQ plots of $\xi_{i,j}$ with black color for $\hat{\xi}_{i,j}$ and red color for $\bar{\xi}_{i,j}$. }
\label{figure-a}
\end{figure}

When $\epsilon_n=2$, the QQ-plots under $n=100, 200, 500$ are similar and
we only show the QQ-plots for $\hat{\xi}_{i,j}$ when $n=100$ in Figure \ref{figure-a} to save space.
The other QQ-plots for $\epsilon_n=\log n/n^{1/4}, \log n/n^{1/2}$ are shown in the online supplementary material.
In the QQ-plots, the horizontal and vertical axes are the theoretical and empirical quantiles, respectively,
and the straight lines correspond to the reference line $y=x$.
In Figure \ref{figure-a}, we first observe that the empirical quantiles agree well with the ones of the standard normality for non denoised estimates (i.e., $\hat{\xi}_{i,j}$) when $L=0$ and $\log(\log n)$, while
there are notable deviations for pair $(1, 2)$ when $L = (\log n)^{1/2}$. These results are very similar to those in \cite{Yan:Leng:Zhu:2016} where
the original bi-degree sequences are used to estimate the parameters.
Second, by comparing the QQ plots for $\hat{\xi}_{i,j}$ (in black color) and $\bar{\xi}_{i,j}$ (in red color), we find that the performance of $\hat{\xi}_{i,j}$
is much better than that of $\bar{\xi}_{i,j}$ for the pair $(n-1, n)$ when $L\ge \log(\log n)$, whose QQ plots derivative from the diagonal line in both ends.
When $\epsilon_n=\log n/n^{1/4}$, the QQ-plots are in Figures 1, 2 and 3 in the online supplementary material, corresponding to $n=100, 200, 500$ respectively.
These figures exhibit similar phenomena.
Moreover, the derivation of the QQ-plots from the straight becomes smaller as $n$ increases, and
they match well when $n=500$.
The QQ-plots under $\epsilon_n=\log n/n^{1/2}$ are drawn in Figures 4, 5 and 6 in the online supplementary material, corresponding to $n=100, 200, 500$ respectively.
In this case, the condition in Theorem \ref{Theorem:binary:central} fails and
these figures shows obvious derivations from the standard normal distribution.
It indicates that $\epsilon_n$ should not go to zero quickly as $n$ increases in order to guarantee good utility.
Lastly, we observe that when $L=\log n$ for which the condition in Theorem \ref{Theorem:binary:central} fails, the estimate did not exist in all repetitions (see Table 1 in supplementary material).
Thus the corresponding QQ plot could not be shown.

In order to assess the effect of the additional variance factor (i.e., $s_n^2/\hat{v}_{2n,2n}^2$) in Theorem \ref{Theorem:binary:central}, we draw
the QQ-plots for $(\hat{\alpha}_{i}-\alpha_i)/\hat{\sigma}_i^{(1)}$ denoted by the black color and $(\hat{\alpha}_{i}-\alpha_i)/\hat{\sigma}_i^{(2)}$ by the red color in Figure 7
in supplementary material, where
$(\hat{\sigma}_i^{(1)})^2=1/\hat{v}_{i,i}+1/\hat{v}_{2n,2n}+s_n^2/\hat{v}_{2n,2n}^2$, $(\hat{\sigma}_i^{(2)})^2=1/\hat{v}_{i,i}+1/\hat{v}_{2n,2n}$, $n=100$ and $\epsilon=2$.
From this figure, we can see that the empirical quantiles agree well with the ones of the standard normality when the variance of $\hat{\alpha}_{i}$ is correctly specified (i.e., $\hat{\sigma}_i^{(1)}$).
When ignoring the additional variance factor,
there are obvious derivations for $(\hat{\alpha}_{i}-\alpha_i)/\hat{\sigma}_i^{(2)}$.
It indicates that the additional variance factor can not be ignored when the noise is not very small,
agreeing with Theorem \ref{Theorem:binary:central}.

Table 1 in supplementary material reports the coverage frequencies of the $95\%$ confidence interval for $\alpha_i - \alpha_j$, the length of the confidence interval, and the frequency that the MLE did not exist.
As expected, the length of the confidence interval increases as $L$ increases and decreases as $n$ increases.
We first look at the simulation results in the case of $\epsilon_n=2$:
when $L\le \log(\log(n))$, most of simulated coverage frequencies for the estimates are close to the targeted level
and the non denoised estimate has better performance than  the denoised estimate;
the values under the pair $(n-1, n)$ corresponding to the denoised estimate are lower than the nominal level when $L=\log(\log(n))$.
When $L = (\log n)^{1/2}$, both denoised and non denoised estimates failed to exist with a positive frequency while
the estimate did not exist in any of the repetitions in the case of $L=\log n$.
The results in the case of $\epsilon_n=\log n/n^{1/4}$ exhibit similar phenomena. However,
the simulated coverage frequencies are a little lower than the nominal level when $n=100$, showing
that smaller $\epsilon_n$ needs larger $n$ to guarantee high accuracy.
The results in the case of $\epsilon_n=\log n/n^{1/2}$ are shown in Table 1 in the online supplementary material.
From this table, we can see that the simulated coverage frequencies are obviously far away from the nominal level and
the estimate fails to exist with positive frequencies when $L\ge \log(\log(n))$.

\subsection{Real data analysis}

We evaluate how close the estimator $(\hat{\alpha}, \hat{\beta})$ is to the MLE $(\tilde{\alpha}, \tilde{\beta})$ fitted in the $p_0$ model
with the original bi-degree sequence through three real network datasets, which
are the Children's Friendship data, Lazega's Law Firm data and Uc irvine messages data, respectively.
We only present the  analytical  results of the Uc irvine messages data here and the others are put in supplementary material. 
Note that $(\hat{\alpha}, \hat{\beta})$ is the edge differentially private estimator of the vector parameters $\alpha$ and $\beta$. If only the
private estimator is released, then whether an edge is present or not in the original dataset could almost not be detected.
We chose $\epsilon_n$ equal to $1$, $2$ and $3$ as in \cite{Karwa:Slakovic:2016} and repeated to release the bi-degree sequence using Algorithm \ref{algorithm:a}
$1,000$ times for each $\epsilon$. Then we computed the average private estimate and the upper ($97.5^{th}$) in blue color and the lower ($2.5^{th}$) quantiles in orange color of the estimates conditional on the event that the private estimate exists.

The  Uc irvine messages network data was collected from an online community of students at the University of California, Irvine [\cite{Opsahl-Panzarasa-2009}].
It has a total of $1899$ nodes and each node represents a student. A directed edge is established from one
student to another if one or more messages have been sent from
the former to the latter. A total of $20,296$ edges form and  the  edge density is $0.56\%$, indicating a very sparse network.
Among $1,899$ nodes, there are $586$ nodes having no out-edges or in-edges. We remove them due to that the non private MLE does not exist in this case.
To guarantee non zero out-degrees and in-degrees after adding noises with a large probability,
we only analyze a subgraph with their out-degrees and in-degrees both larger than $5$.
After data preprocessing, only $696$ nodes are left and
the quantiles of $0$, $1/4$, $1/2$, $3/4$, $1$ are  $3$, $8$, $14$, $26$, $164$ for out-degrees and $4$, $10$, $16$, $27$, $121$
for in-degrees, respectively.

When many nodes have few links to others,  large noise is easy to cause the output with non positive elements in Algorithm \ref{algorithm:a}.
When $\epsilon=1$, the average $\ell_\infty$-distance between $d$ and $\tilde{d}$ is $15.6$
and all private estimates fail to exist. In this case, we try another $\epsilon=\log n/n^{1/4}$ ($\approx 1.27$).
The frequencies that the private estimate fails to exist are $99.3\%$, $54.9\%$ and $8.3\%$ for $\epsilon=\log n/n^{1/4},2,3$, respectively.
The results are shown in Figure \ref{figure-message}. From this figure, we can see that the mean value of $\hat{\alpha}$ or $\hat{\beta}$
are very close to the MLE and the MLE still lies in the $95\%$ confidence interval.

\begin{figure}[!htb]
\centering
\includegraphics[ height=4in, width=6in, angle=0]{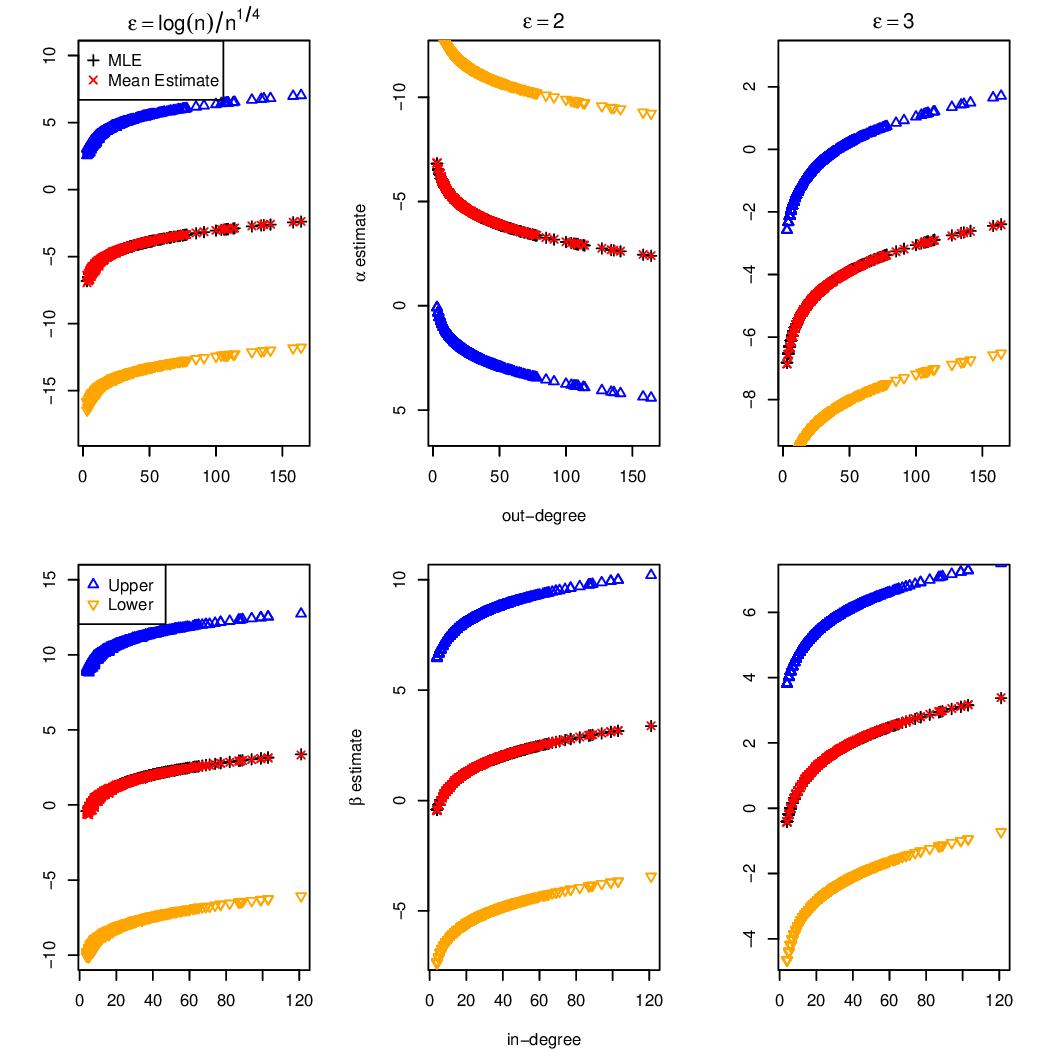}
\caption{The differentially private estimate $(\hat{\alpha}, \hat{\beta})$ with the MLE for the Uc irvine messages network.
}
\label{figure-message}
\end{figure}

\section{Discussion}
\label{section:discussion}

We have presented the consistency of the differentially private estimator of the parameter in the $p_0$ model under some mild conditions when the
discrete Laplace noise is added into the bi-degree.
We have revealed a phase transition for the asymptotic variance of the estimator in which an additional variance factor appears when
the variance of the noise increases. The simulation shows that ignoring it could lead to invalid conference intervals.
The added noise introduces considerable error when applying the noisy bi-sequence to estimate the degree distribution.
We propose an efficient algorithm to denoise the noisy bi-sequence.
The denoised bi-sequence can be used to obtain an accurate estimate of the degree distribution of a directed graph.
Our simulation studies show that the non denoised estimator has a better performance than the denoised estimator for finite network sizes.
On the other hand, when the privacy parameter $\epsilon_n$ is small, the private estimate fails to exist with
positive frequencies according to simulations and real data analyses, especially when the network dataset is sparse.
An approach to avoid this problem is adding positive Laplace random noises or using $f$-differential privacy.
We would like to investigate this problem in the future.

The conditions in Theorems \ref{Theorem:con} and \ref{Theorem:binary:central} induce an interesting trade-off between the private parameter
measuring the magnitude of the noise 
and the growing rate of the parameter $\theta$.
If the parameter $\epsilon_n$ is large, $\theta$ can be allowed to be relatively large.
For instance,  if $\epsilon_n = O(1)$, then the condition (i.e., $(1+4\epsilon_n^{-1}) e^{12\|\theta^*\|_\infty } = o( (n/\log n)^{1/2} )$) in Theorem \ref{Theorem:con}
becomes $e^{12\|\theta^*\|_\infty } = o( (n/\log n)^{1/2} )$.
Moreover, the condition in Theorem \ref{Theorem:binary:central} is much stronger than that in Theorem \ref{Theorem:con}.
The asymptotic behavior of the estimator is not only determined by the growing rate of the parameter $\theta$, but also by the configuration of the parameter.
It would be of interest to see whether these conditions can be relaxed.

There are two different tasks for data privacy problem. The first is data protection.
If the network model contains other network features such as
$k$-stars and triangle and only these network statistics are of interest, then
the additive noisy mechanism in this paper can be used to disclose them safely and it satisfied the edge differential privacy if the Laplace noise is added.
The second is making inference from the noisy data.
In order to extend the method of deriving the consistency of the estimator in our paper to other network models,
one needs to establish a geometrical rate of convergence of the Newton iterative sequence.
This is not easy for network models with other network features since it is difficult to derive the upper bound of the matrix norm for
the inverse matrix of the Fisher information matrix without some special matrix structures.
At the same time, it is also difficult to extend the method of deriving asymptotic normality of the estimator to network models with other network features
since it is generally difficult to derive the approximate inverse matrix of a general Fisher information matrix.


\end{document}